\documentclass[11pt]{article}
\usepackage[margin=1in]{geometry}
\usepackage{amsfonts,amsmath,amssymb,amsthm}
\usepackage{enumerate}
\usepackage{hyperref}
\usepackage{paralist}

\usepackage{privmath,privref, privalgorithm}
\usepackage[long]{pwbib}
\usepackage{cite}
\usepackage{pgfplots}

\newcommand{\Exp}{\ensuremath{\operatorname{\mathbf{E}}}}

\usepackage[textsize=tiny,disable]{todonotes}

\SetKwFunction{reset}{reset}
\SetKwFunction{scan}{scan}
\SetKwFunction{TAS}{test\&set}
\SetKwFunction{win}{win}
\SetKwFunction{lose}{lose}
\SetKwFunction{continue}{continue}
\SetKwFunction{True}{true}
\SetKwFunction{False}{false}
\SetKwFunction{Knockout}{knockout}
\SetKwFunction{Competition}{compete}
\SetKwFunction{Enter}{enter}
\SetKwFunction{split}{split}
\SetKwFunction{compete}{compete}
\SetKwFunction{op}{op}

\renewcommand{\index}{\ensuremath{\mathrm{index}}}
\newcommand{\pos}{\ensuremath{\mathrm{pos}}}
\newcommand{\sig}{\ensuremath{\mathrm{sig}}}

\newcommand{\id}{\ensuremath{\mathrm{id}}}
\newcommand{\num}[2]{\ensuremath{\mathrm{num}(#1,#2)}}

\newcommand{\procs}{\ensuremath{\mathcal{P}}}
\newcommand{\PP}{\ensuremath{\mathcal{P}}}

\newtheorem{theorem}{Theorem}
\newtheorem{lemma}[theorem]{Lemma}
\newtheorem{observation}[theorem]{Observation}
\newtheorem{corollary}[theorem]{Corollary}

\title{Deterministic and Fast Randomized Test-and-Set in Optimal Space \footnote{The results in this paper combine and elaborate on previous work that appeared in the
2015 ACM Symposium on Theory of Computing~\cite{GHHW2015} and 2013 International Symposium on Distributed Computing~\cite{GHHW2013a}}}
\author{George Giakkoupis\\
INRIA Rennes\\
\normalsize{george.giakkoupis@inria.fr}
\and
Maryam Helmi\\
University of Calgary\\
\normalsize{mhelmikh@ucalgary.ca}
\and
Lisa Higham\\
University of Calgary\\
\normalsize{higham@ucalgary.ca}
\and
Philipp Woelfel\\
University of Calgary\\
\normalsize{woelfel@ucalgary.ca}
}

\date{}

\begin{document}
\maketitle

\begin{abstract}
The test-and-set object is a fundamental synchronization primitive for shared memory systems.
A test-and-set object stores a bit, initialized to 0, and supports one operation, \TAS{}, which sets the bit's value to 1 and returns its previous value.
This paper studies the number of atomic registers
required to implement a test-and-set object in the standard asynchronous shared memory model with $n$ processes.
The best lower bound is $\log n - 1$ for obstruction-free~\cite{GW2012b} and deadlock-free~\cite{SP1989a} implementations.
Recently a deterministic obstruction-free implementation using $O(\sqrt{n})$ registers was presented \cite{GHHW2013a}.
This paper closes the gap between these known  upper and lower bounds by presenting
a deterministic obstruction-free implementation of a test-and-set object from
$\Theta(\log n)$ registers of size $\Theta(\log n)$ bits.

We also provide a technique to transform any deterministic obstruction-free algorithm, in which, from any configuration, any process can finish if it runs for $b$ steps without interference, into a randomized wait-free algorithm for the oblivious adversary, in which the expected step complexity is polynomial in $n$ and $b$. This transformation allows us to combine our obstruction-free algorithm with the randomized test-and-set algorithm by Giakkoupis and Woelfel~\cite{GW2012b},
to obtain a randomized wait-free test-and-set algorithm from $\Theta(\log n)$ registers, with expected step-complexity $\Theta(\log^\ast n)$ against the oblivious adversary.

\end{abstract}

\addtocounter{page}{-1}
\thispagestyle{empty}
\clearpage

\section{Introduction}

A \emph{test-and-set} (\emph{TAS}) object is perhaps the simplest standard shared memory primitive that has no wait-free deterministic implementation from registers.
It stores a bit, which is initially 0, and supports one operation, namely \TAS{}.
A  \TAS{} sets the bit's value to 1 and returns its previous value.

TAS objects have consensus number two.
That is, they can be used together with registers to solve deterministic wait-free consensus only in systems with two processes.
Despite that, TAS is a standard building block for shared memory algorithms
that solve many classical problems, such as mutual exclusion and renaming~\cite{KRS1988a,PPTV1998a,EHW1998a,BPSV2006a,AAGGG2010a,AAGG2011a,AACGZ2011a}.
Since TAS objects are among the simplest synchronization primitives, they are well suited for investigating the difficulties arising in synchronization problems.
Algorithms or impossibility results for TAS provide insights into the complexity of other shared memory problems, and can contribute to their solutions.

We consider a standard shared memory system in which $n$ processes communicate through atomic read and write operations on shared registers.
A common assumption is that each register can store $\Theta(\log n)$ bits, although in some settings registers of unbounded size are assumed.
The strongest reasonable progress condition is \emph{wait-freedom}, which guarantees that every operation finishes in a finite number of the calling process' steps, independent of other processes.
Since TAS has consensus number two, deterministic wait-free implementations from registers do not exist for two or more processes.
A weaker progress condition, and the one most frequently used for analyzing space complexity, is \emph{obstruction-freedom} \cite{HLM2003a}.
It guarantees that from every reachable configuration and for any process, that process will finish its operation in a finite number of its own steps, provided that no other process takes any steps (i.e., in a sufficiently long solo execution).
Any shared memory object has an obstruction-free implementation from $n$ registers~\cite{Her1991a}.

The randomized step complexity of TAS has been thoroughly investigated, with significant progress being made in recent years \cite{TV2002a,AGTV1992a,AAGGG2010a,AA2011a,GW2012a}.
In contrast, little was known about the space complexity of obstruction-free or randomized wait-free TAS.
In 1989, Styer and Peterson~\cite{SP1989a} studied the space complexity of the related mutual exclusion problem, 
under the deadlock-free progress requirement.
As a special case they also considered a variant called weak leader election (see Section \ref{sec:model}). 
It suffices to add a single one-bit register to transform any deadlock-free weak leader election protocol 
into a linearizable deadlock-free TAS.
Styer and Peterson proved a space lower bound of $\ceil{\log n}+1$ registers,
and provided an algorithm that established that this bound is tight.
Hence, in the case of deadlock-freedom, Styer and Peterson's results answer the question of the space complexity of TAS precisely up to a single register.

Deadlock-freedom is a natural progress property for mutual exclusion related problems,
where waiting for other processes is inherent in the problem specification.
But for other problems, it is inappropriate because it does not preclude a single slow or failing process preventing all other processes from making progress.
Alternative progress properties, such as obstruction-freedom, lock-freedom, or randomized wait-freedom,
are more desirable for such problems.
Research on the space complexity of shared memory problems has focused on the obstruction-free progress property \cite{JTT2000a,ELMS05,HLM2003a,GHP05}.
However, despite significant research on TAS, prior to the result presented here, the asymptotic space complexity of obstruction-free TAS implementations remained unknown.

In 2012, Giakkoupis and Woelfel~\cite{GW2012b} used the same lower bound technique as that of Styer and Peterson to conclude that obstruction-free TAS
requires $\ceil{\log n} - 1$ registers.
The maximum number of steps taken by any process running alone, until it finishes its method call is called the \emph{solo step complexity} of that method~\cite{AGHK09}.
In 2013 we devised a deterministic obstruction-free TAS algorithm using $\Theta(\sqrt n)$ registers, where the solo step complexity of \TAS{} is $\Theta(\sqrt n)$~\cite{GHHW2013a}.
We now present an asymptotically tight result.

\begin{theorem}\label{thm:main}
  There is a deterministic obstruction-free implementation of a TAS object from $\Theta(\log n)$ registers of size $\Theta(\log n)$ bits,
  where the solo step complexity of the \TAS{} method is $\Theta(\log n)$.
\end{theorem}

There are performance benefits if the solo run that is required for termination is short,
because processes have a better chance of completing their method call before they get interrupted.
In our algorithm, processes make partial progress even if they can run uninterruptedly for a constant number of steps. As a result, a process needs to execute only a constant number of solo steps $\Theta(\log n)$ times, to finish its \TAS{} method call.

The relation between wait-freedom and obstruction-freedom has been investigated before: Fich, Luchangco, Moir, and Shavit~\cite{ELMS05} showed that obstruction-free algorithms can be transformed into wait-free ones in the unknown-bound semi-synchronous model.
The approach in this paper is different;
 we use randomization,
but stay in the fully asynchronous model.
It is easy to see that any deterministic obstruction-free algorithm can be transformed into an algorithm that is randomized wait-free against the oblivious adversary and has exponential expected step complexity.
In Section \ref{sec:wait-free}, we provide a more efficient but also simple transformation to show the following result.
\begin{theorem}
    \label{thm:randomized-step-bound}
    Suppose there is a deterministic obstruction-free algorithm whose solo step complexity is $b$.
    Then the algorithm can be transformed into a randomized one that uses the same number of registers of the same size, such that for any schedule
    determined by an oblivious adversary, each process finishes after at most $O\big(b(n + b) \log(n/\delta)\big)$ of its own steps with probability at least $1-\delta$, for any $\delta > 0$
    (which can be a function of $n$).
\end{theorem}

We apply this transformation to our obstruction-free algorithm and combine the result with the test-and-set algorithm by Giakkoupis and Woelfel~\cite{GW2012b}, to obtain a randomized wait-free TAS implementation from $\Theta(\log n)$ registers, which has expected step complexity $O(\log^\ast n)$.


\begin{theorem}\label{thm:randomized}
  There is a randomized TAS implementation from $\Theta(\log n)$ registers of size $\Theta(\log n)$ bits, such that for any schedule determined by an oblivious adversary, the maximum number of steps executed by any process is $O(\log^\ast n)$ in expectation, and $O(\log n)$ with high probability, $1 - n^{-\Omega(1)}$.
\end{theorem}

A \emph{long-lived} test-and set object provides an operation \reset{} in addition to \TAS{}.
The \reset{} operation can only be executed by a process if its preceding operation on the object was a successful \TAS{}; in that case the \reset{} operation unconditionally resets the value of the TAS object to 0.
Recently, Aghazadeh and Woelfel \cite{AW2014b} showed that any TAS object implemented from $m$ $\ell$-bit registers can be transformed into a long-lived TAS object, using $O(m\cdot n)$ registers of size $\max\{\ell,\,\log(n+m)\}+O(1)$ bits.
A \reset{} operation takes only constant time in the worst-case, and the step complexity of a \TAS{} operation of the long-lived object is the same (up to a constant additive term) as the one of the (one-shot) TAS object.
Applying this to the result stated in Theorem~\ref{thm:randomized}, yields the following:
\begin{corollary}
 A long-lived TAS object can be implemented from $O(n\log n)$ registers, each of size $O(\log n)$ bits, such that the expected step complexity of \TAS{} is $O(\log^\ast n)$ against the oblivious adversary, and the worst-case step complexity of \reset{} is $O(1)$.
\end{corollary}
The space lower bound for mutual exclusion \cite{BL1993a} implies that any long-lived TAS implementation requires at least $n$ registers.
Aghazadeh and Woelfel \cite{AW2014b} also gave a construction of a long-lived TAS from $O(n)$ registers, where the expected step complexity of \TAS{} and \reset{} is $O(\log\log n)$ against the oblivious adversary.

Our TAS algorithms rely on two components that are of independent interest; we expect they have other applications.
One is  an $M$-component snapshot object implemented from bounded registers.
A $B$-bounded $M$-component snapshot object maintains a collection of $M$ components.
Each component stores a value of size at most $B$ bits.
The object supports two operations $\update{i,x}$ and $\scan{}$.
Operation \scan{} returns the values of all components, and $\update{i,x}$ writes $x$ to the $i$-th component where $x$ has size at most $B$ bits and $i\in\{1,\dots,M\}$.
If each component has unbounded size, then it is simply called an $M$-component snapshot object.
The snapshot object is an important and well-studied primitive in distributed computing.
There are many implementations of snapshot objects from registers in the literature~\cite{AADGMS1993a,A1994,AA2002,AR1998,FFR2007a}.
The lower bound by Jayanti, Tan and Toueg for the general class of perturbable objects
implies that any implementation of an $M$-component snapshot object from historyless and resettable consensus objects requires at least $M - 1$ objects,
and each \scan{} operation takes at least $M - 1$ steps~\cite{JTT2000a}.
Fatourou, Fich and Ruppert improved the space lower bound for $M$-component snapshot objects to $M$ for implementations from registers~\cite{FFR2007a}.
They also showed that this lower bound is tight by providing a wait-free implementation of an $M$-component snapshot object from $M$ unbounded registers. 
But for our test-and-set implementation we need an asymptotically optimally space efficient snapshot object that uses only bounded registers.
Section \ref{sec:snapshot} contains our simple obstruction-free implementation of a $B$-bounded $M$-component snapshot object from $M+1$ bounded registers.

\begin{theorem}\label{thm:scan}
There is an obstruction-free implementation of a $B$-bounded $M$-component snapshot object from $M+1$ registers of size $\Theta(B +\log {n})$ bits,
where the solo step complexity of \scan{} is $O(M)$ and the solo step complexity of \update{} is $O(1)$.
\end{theorem}

The key component of our TAS algorithm is a sifter object.
An \emph{$f(k)$-sifter}, where $f$ is a function such that $1\leq f(k)\leq\max\{k-1,1\}$ for any integer $k\geq1$, supports only one method, \Competition{}, which returns \win or \lose.
In any execution where $k$ processes call \Competition{}, at most $f(k)$ of them return \win, and at most $k-1$ return \lose.
Recent randomized TAS constructions~\cite{AA2011a,GW2012a} are based on randomized sifters,
where the number of winning processes is at most $f(k)$ in \emph{expectation}.
Here, however, we use deterministic sifters, where $f(k)$ is a worst-case bound.
Section~\ref{sec:sifter} contains our sifter implementation, which establishes the following theorem.
\begin{theorem}\label{thm:sifter}
	There is an obstruction-free implementation of a $\floor {\frac{2k+1}{3}}$-sifter
	from a $\ceil{4 \cdot\log {n}}$-bounded $6$-component snapshot object.
\end{theorem}
By combining $O(\log n)$ sifters, and our snapshot object from Theorem~\ref{thm:scan} we obtain our TAS implementation using $O(\log n)$ registers.


Section~\ref{sec:model} defines our model of computation and communication.
We assume these sifter and snapshot tools to implement our deterministic TAS object in Section~\ref{sec:deterministicTAS},
and, together with Theorem~\ref{thm:randomized-step-bound}, our randomized TAS object in Section~\ref{sec:randomizedTAS}.

\section{Model and Preliminaries}
\label{sec:model}

Our model of computation and communication is the standard asynchronous shared memory model where a set $\procs$ of $n$ processes
with distinct identifiers communicate through shared multi-reader multi-writer registers.
Each register supports two atomic operations, read and write.

An algorithm is an assignment of a program to each process.
Each process' program can access that process' 
local registers as well as the shared registers.
At each \emph{step} by a process, that process executes
a single shared memory access (or, initially, its program invocation)
followed by all its subsequent local operations and random choices,
up to the point where that process is poised to execute its next shared memory operation.
A \emph{schedule} is a sequence of process identifiers.
A schedule, $\sigma$, gives rise to a sequence of steps, called an \emph{execution} as follows.
The $i$-th step in the execution is the next step in the program of the $i$-th process in $\sigma$.

An algorithm is \emph{deterministic} if each process' program is deterministic.
A deterministic implementation of a method is \emph{wait-free} if,
from any point of an execution and for any process,
the process completes its method call in a finite number of its own steps,
regardless of the intervening steps taken by other processes.
A deterministic implementation of a method is \emph{obstruction-free} if,
from any point of an execution and for any process $p$,
$p$ completes its method call in a finite number of its own steps,
provided there are no intervening steps taken by other processes.
In such an execution, we say that $p$ runs \emph{solo} during these uninterrupted steps by $p$.

The algorithm is \emph{randomized} if some process' program is randomized.
An implementation of a method is \emph{randomized wait-free} if,
from any point of an execution and for any process $p$,
the  number of steps by $p$ required for $p$ to complete its method call is finite in expectation,
regardless of the intervening steps taken by other processes~\cite{Her1991a}.

A \emph{test-and-set} (\emph{TAS}) object stores one bit, which is initially 0,
and supports a \TAS{} operation that sets the bit's value to 1 and returns its previous value.

An \emph{$f(k)$-sifter} object, where $f$ is a function such that $1\leq f(k)\leq\max\{k-1,1\}$ for any integer $k\geq1$, supports only one operation, \Competition{}, which returns \win or \lose.
In any execution where $k$ processes call \Competition{}, at most $f(k)$ of them return \win, and at most $k-1$ return \lose.

A $B$-bounded $M$-component snapshot object stores a vector $V=(V_1,\dots,V_M)$ of $M$ values from some domain $D$,
where each $d$ in $D$ has size at most $B$ bits.
It supports two operations: \scan{} takes no parameter and returns the value of $V$,
and $\update{i,x}$, $i\in\{1,\dots,M\}$, $x\in D$, writes $x$ to the $i$-th component of $V$ and returns nothing.

An object is \emph{implemented} by providing a program, the \emph{method} for \op,
for each operation, \op, defined for that object.
Since our objective is to implement a TAS object, we need to provide a \TAS method.
Our TAS algorithm is then just the \TAS method assigned to each process.
Our correctness condition is \emph{linearizability} \cite{HW1990a}, which requires that for any execution of our algorithm
and for every \TAS method call, $\mu$, in that execution, there is a point between $\mu$'s invocation and response
such that if the entire method call is replaced by the atomic \TAS operation returning the same value as $\mu$ at that point,
the resulting execution is valid for the TAS object.
Linerizability is a \emph{composable} property:
A linerizable implementation of object $A$ assuming atomic objects $B$,
composed with a linearizable implementation of $B$ assuming atomic objects $C$,
is a linearizable implementation of $A$ using $C$.
We exploit this by providing a linearizable implementation of a TAS object assuming an $M$-component snapshot object,
and then a linearizable implementation of an $M$-component snapshot on registers.
Linearizabilty is also a \emph{local} property:
any correct deterministic algorithm that uses a collection of atomic objects,
will remain correct if these objects are replaced with their linearizable implementations.

Implementing a TAS object is related to solving \emph{weak leader election}, 
where each participating process has to decide on one value, \win or \lose.
Among all processes that finish their weak leader election protocol, at most one process is allowed to win, and not all processes may lose.
Hence, if all processes finish, then exactly one process, the leader, wins.
(The term leader election is ambiguous.
It is also used to denote the \emph{name consensus} problem,
where the losing processes need to output the ID of the winner.
We add the qualifier ``weak'' in order to distinguish the two variants.)
Weak leader election and test-and-set are equally hard problems with respect to asymptotic space complexity.
Replacing the return values 0 and 1 of a \TAS{} operation with \win and \lose, respectively, yields a weak leader election protocol.
The difference is that TAS requires that the \TAS method that  returns 1 must be linearized before those that return 0,
whereas weak leader election lacks the corresponding requirement for \win and \lose.
Nevertheless, Golab, Hendler and Woelfel~\cite{GHW2010a} gave an implementation of a TAS object
using weak leader election and one additional register:

\begin{theorem}\label{thm:TAS}
	\cite{GHW2010a} A linearizable TAS object can be implemented using a weak leader election protocol and one additional multi-reader/multi-writer binary register, such that a \TAS{} method requires only a constant number of read and write operations in addition to the weak leader election protocol.
\end{theorem}

For a deterministic obstruction-free implementation, the \emph{solo step complexity} is the worst case
over all processes $p$ and all reachable configurations $C$ of the number of steps taken in a solo execution by $p$
starting at $C$ until $p$ terminates its method.

\section{Space Efficient Deterministic Test-and-Set}\label{sec:deterministicTAS}

Because of Theorem \ref{thm:TAS}, to establish Theorem \ref{thm:main},  
it suffices to give an implementation of weak leader election that achieves the space and step complexity claimed in that theorem. 
We now describe this implementation, assuming we have the use of the sifter object of Theorem \ref{thm:sifter} and the snapshot object of Theorem \ref{thm:scan}.

An $f(k)$-sifter and a $g(k)$-sifter
can be combined to obtain an $f(g(k))$-sifter, by letting the losers of the $g(k)$-sifter lose, and the winners call \Competition{} on the $f(k)$-sifter.
Hence, by combining enough sifter objects, we can obtain a 1-sifter, which is a weak leader election protocol.

In Section \ref{sec:sifter} we show how to implement a single $\floor {\frac{2k+1}{3}}$-sifter
from a 6-component snapshot object.
The implementation is obstruction-free.
Moreover, whenever a process starts running alone, it terminates after $O(1)$ scan and write operations.
By Theorem~\ref{thm:scan}, we can implement a 6-component snapshot object from $7$ registers,
where the  solo step complexity of each method is constant.
Hence, using the obstruction-free snapshot implementation from Theorem~\ref{thm:scan},
our $\floor {\frac{2k+1}{3}}$-sifter implementation has constant solo step complexity and uses $7$ registers.

Since multiple sifters are combined to construct our weak leader election algorithm (and hence our TAS implementation)
it is more space efficient to replace the individual snapshot objects with a single snapshot object shared by all sifters.
We can simulate $\ell$ distinct 6-component snapshot objects by one $(6\ell)$-component snapshot object.
By Theorem~\ref{thm:scan}, we can implement such a snapshot object using $6 \ell + 1$ registers where the solo step complexity is $O(\ell)$.
Hence, Theorem~\ref{thm:scan} and Theorem~\ref{thm:sifter} combine to yield:
\begin{corollary}\label{cor:combinedSifters}
  There is an obstruction-free implementation of $\ell$ instances of $\floor {\frac{2k+1}{3}}$-sifters using $6 \ell + 1$ registers,
  each of size $\Theta(\log n)$-bits, such that the solo step complexity of \Competition{} is $O(\ell)$.
\end{corollary}

We can implement a weak leader election protocol using a sequence of at most $\ell = \floor{\log_{3/2} n} + 1$ instances of a $\floor {\frac{2k+1}{3}}$-sifter.
As describe earlier, each process starts by invoking the \Competition{} method of the first sifter;
the winners of the $i$-th sifter proceed to the $(i+1)$-th sifter,
while the losers lose the weak leader election;
the winner of the weak leader election is the process that wins the last sifter.
We need to show that $\ell$ repeated applications of function $f(k) = \floor {\frac{2k+1}{3}}$ to an initial value of $k=n$ yield a value of 1.


\begin{lemma}\label{lem:repeated-log}
	Let $f(n) = \floor{\frac{2n+1}{3}} $.
	Let $f^{(0)}(n) = n$ and $f^{(i+1)}(n) = f(f^{(i)}(n))$.
	Then for any  integer $\ell \geq \log_{3/2} n $,  $f^{(\ell)}(n) = 1$ for any $n \geq 1$.
\end{lemma}
\begin{proof}
	First, observe that if $n\geq 1$ then $\floor{\frac{2n+1}{3}} \geq 1$, so $f^{k}(n)$ never drops below $1$ for any $k$. 	
	Now, we show by induction on $k$, that
	\begin{displaymath}
	f^{(k)}(n) \leq \left(\frac{2}{3} \right)^k  n + 1 - \left(\frac{2}{3} \right)^k \text{ for }  k \geq 0.
	\end{displaymath}
	For the basis, $k=0$, observe that $f^{(0)}(n) = n = (\frac{2}{3} )^0 \cdot n + 1 - (\frac{2}{3})^0$.\\
	For the inductive step:
	\begin{eqnarray*}
	f^{(k+1)}(n) &=&   f(f^{(k)}(n)) = \left\lfloor \frac{2(f^{(k)}(n))+1}{3} \right\rfloor \leq \frac{2(f^{(k)}(n))+1}{3} \\
	           &\leq & \frac{2 \big( (\frac{2}{3} )^k  n + 1 - (\frac{2}{3})^k \big) + 1}{3}  \text { by the induction hypothesis} \\
	           &=& 	\left(\frac{2}{3}\right)^{k+1} n + 1 -\left(\frac{2}{3}\right)^{k+1}.
	\end{eqnarray*}
	
Thus, for any integer $ \ell \geq \log_{3/2} n$, $f^{(\ell)} (n) \leq 1 +1 - \frac{1}{n} <2$.

But $f$ takes only integer values and $\ell$ is an integer,
implying that after $\floor {\log_{3/2} n} +1$ applications of $f$,
the value is at most 1.	
\end{proof}

Thus, Theorem~\ref{thm:main} follows from  Theorem \ref{thm:TAS}, Corollary~\ref{cor:combinedSifters} and Lemma~\ref{lem:repeated-log}.
More precisely, we have:
\begin{theorem}\label{thm:tightenedMain}
	There is a deterministic obstruction-free implementation of a TAS object from
	$6 \floor{\log_{3/2} n} + 7 $ registers each of size at most $ 4\log n$ bits,
	where the solo step complexity of the \TAS{} method is $\Theta(\log n)$.
\end{theorem}
Since Corollary~\ref{cor:combinedSifters} follows from Theorem~\ref{thm:sifter} and Theorem~\ref{thm:scan},
it remains to prove these two theorems to complete the implementation of our deterministic TAS object.
This we do in the next two sections.

\section{Sifter Implementation}\label{sec:sifter}

\begin{figure*}
  \renewcommand{\algorithmcfname}{Algorithm}
  \DontPrintSemicolon
  \begin{center}
  \begin{minipage}{.9\textwidth}
	\textbf{Shared Objects:}
	\begin{itemize}
		\item[$\bullet$] $A[0, 1, 2]$ is an array of the first 3 components of a 6-component snapshot object $U$. Each array entry stores a value from $\procs \cup \{\bot\} $ and is initially $\bot$.
\smallskip

		\item[$\bullet$] $B[0, 1, 2]$ is an array of the second 3 components of $U$. Each array entry stores a pair $(\id,\sig)$, where $\id \in \procs \cup \{\bot\} $, and $\sig$ is a triple from the set $(\procs \cup \{\bot\} )^3$. Initially, $\id = \bot$ and $\sig = (\bot,\bot,\bot)$.
	\end{itemize}
	
	\textbf{Notation:} For any array $X$ and value $v$, let $\num{v}{X} :=|\{i:X[i] = v\}|$.
	\bigskip
	
	\begin{algorithm}[H]
    \nonl\TitleOfAlgo{compete()}
		
		$\pos := 0$\;
		\While(\label{line:while}){\True}{
			$A[\pos]$.write$(p)$\label{line:Awrite} \;
			$a:=$ scan$(A)$\label{line:Cscan}\;	
			\IlIf{$\num{p}{a}=3$
\label{line:Awin-condition}}{\Return{\win}\label{line:Awin}}\;
			\IlIf{$\exists\, q \in \procs\colon \num{p}{a}<\num{q}{a}$ \label{line:second-if}}{\Return{\lose}\label{line:first-die}}\;
			\If{$\num{p}{a} = 1$\label{line:num1}}{	
				\IlIf{\Knockout{a}\label{line:invokeK}}{\Return{\lose}\label{line:second-die}}			
			}
			Let $\pos \in\set{0,1,2}\colon a[\pos] \neq p$ \KwSty{and} $a[(\pos - 1)\bmod 3] = p$ \label{line:pos}\;
		}
	\end{algorithm}
  \setlength{\interspacetitleruled}{0pt}%
  \setlength{\algotitleheightrule}{0pt}
  \bigskip
  
	\begin{function}[H]
	\nonl\TitleOfAlgo{knockout($\sig$)}
		$\index := 0$\;
		\While{\True}{
			$B[\index]$.write$((p,\sig))$\label{line:Bwrite}\;
			$(\widehat{a},\widehat{b}):=$ scan$(A,B)$\label{line:Kscan}\;
			\IlIf{$\widehat{a} \neq \sig$\label{line:A-changed}}{\Return{\True}}\label{line:first-true}\;
			\IlIf{$\exists\, q\in\PP\colon q \neq p$ \KwSty{and} $\num{(q,\sig)}{\widehat{b}} \geq 2$}{\Return{\True}}\label{line:second-true}\;
			\IlIf{$\num{(p,\sig)}{\widehat{b}}=3$}{\Return{\False}}\label{line:Bwin}\;
			Let $\index \in\set{0,1,2}\colon \widehat{b}[\index] \neq (p,\sig)$\label{line:index}\;
		}
	\end{function}
	\caption{Implementation of a sifter for process $p \in \procs$}
	\label{fig:Tas-alg}
  \end{minipage}
  \end{center}
\end{figure*}

This section establishes Theorem \ref{thm:sifter}.
Our sifter implementation is presented in Figure~\ref{fig:Tas-alg}.
To aid intuition we first consider a very simple obstruction-free sifter object, implemented from a 3-component snapshot object $A$.
Each component of $A$ can hold one process identifier.
For ease of readability, we write $A[i]$.write$(x)$ instead of $A$.\update{i,x}, and call \update{} operations writes.
The \scan{} operation returns a triple of process identifiers, called a \emph{signature}.
At some point in an execution, process $p$ \emph{covers} component $i$
if it writes to component $i$ in its next step.
Each process $p$ alternates between writing and scanning.
When $p$ writes, it writes its own identifier to a component of $A$
that did not contain $p$ in its preceding scan.
The goal of any process, $p$, is to achieve a \emph{clean-sweep}
meaning that its scan returns signature $(p,p,p)$.
In this case, $p$ terminates with \win.
If, however, while trying for a clean-sweep,
$p$'s scan returns a signature that contains more copies of a different identifier than it has copies of $p$,
then $p$ terminates with \lose.
Any process that runs alone for six steps without losing, will return \win.
Furthermore, not all processes can return \lose.
To see this, let $w$ be the last write to $A$ and let $p$ be the process executing $w$.
If process $p$ returns \lose,
then there is a process $q$  that occupies two positions in $p$'s last scan, so $q$ cannot return \lose.
Therefore, this is an implementation of an obstruction-free sifter object.

This implementation, however, is not a very efficient sifter.
Suppose that while a clean-sweep is being achieved by one process,
two other processes cover two distinct components of $A$.
Then these covering processes can over-write the clean-sweep,
and be made to again cover two distinct components.
Now a new process can run under the cover and achieve a clean sweep.
By repeating this scenario,
executions are easily created where all but one process return \win.
Also, notice that to create another winner after a clean-sweep,
such an obliteration of the clean-sweep by two (or three) over-writes is also necessary.

To reduce the number of processes that can return \win to at most a constant fraction of those that compete,
the core idea is to
prevent processes that participate in over-writing a clean-sweep, from covering again,
without some process losing.
This is achieved, in our algorithm,
by expanding the 3-component snapshot object $A$ with 3 additional components.
The first 3 components are referred to as $A$, and the second 3 components as $B$.
We implement $A$ and $B$ together from a 6-component snapshot object $U$.
To make notation more intuitive we use the following convention: for each $i \in \{0,1,2\}$, $A[i]$.write$(x)$ denotes $U$.\update{i,x} and
$B[i]$.write$(x)$ denotes \mbox{$U$.\update{i + 3,x}}.
Furthermore, scan$(A,B)$ returns simply what $U$.\scan{} returns, and
scan$(A)$ returns the first three components returned by $U$.\scan{}.

Each component of $B$ can hold a pair consisting of a process identifier and a signature.
A write by $p$ can be either a write of $p$ to a component of $A$,
or a write of $(p, s_p)$ to a component of $B$, where $s_p$ is a signature.
Each process begins by competing on $A$ and still strictly alternates between writing and scanning.

If process $p$, competing on $A$, gets a scan with signature $s$ of $A$,
where the identifiers in $s$ are all distinct and one of them is $p$,
then $p$ leaves $A$ to compete on $B$ while remembering $s$.
(Notice that if $p$ does not get such a scan and it does not immediately return \lose,
then $p$ is in at least two positions in $s$.
Therefore, its last write could not have been part of an over-write of a clean-sweep by some other process.)
By writing the pair $(p,s)$ to components of $B$,
$p$ tries to achieve a clean-sweep of $B$
(meaning a scan by $p$ shows that each of the 3 components of $B$ contains $(p,s)$).
If $p$ achieves such a clean-sweep, then it returns to competing on $A$, as described above.
There are two ways that process $p$ can lose while playing on $B$.
First, $p$ loses if, while trying to achieve a clean-sweep of $B$,
one of $p$'s scans shows a signature of $A$ different from $s$.
Second, $p$ loses if its scan shows that for some other process $q$,
$(q,s)$ occupies at least 2 positions of $B$.
That is, $p$ only returns to continue competing on $A$ if it achieves a clean-sweep of $B$
while each of its scans satisfies
1) the signature of $A$ is $s$,
and 2) no other process with signature $s$ occupies more than one component of $B$.

\subsection{Intuition for Correctness}

Our proof will establish that not all processes can return \lose,
and at most $\floor{(2k+1)/3}$ processes can win, if $k$ processes participate.
While the proof has to attend to several subtleties and substantial detail,
there are several insights that aid our intuition.
We say a process is playing on $A$, if its next shared memory step is on $A$.
Consider the three ways that a process can return \lose.
Let us say \emph{$p$ loses on $A$} if process $p$ loses while playing on $A$
because the signature of $A$ in its last scan contained more occurrences
of some other process than occurrences of $p$.
We say \emph{$p$ signature-loses on $B$} if process $p$ with signature $s$,
loses while trying to achieve a clean-sweep of $B$,
because one of $p$'s scans shows a signature of $A$ different from $s$.
We say \emph{$p$ process-loses on $B$} if process $p$ loses
because its scan shows that for some other process $q$,
$(q,s)$ occupies at least 2 positions of $B$.

Lemma~\ref{lem:lose} below states that not all processes can lose. For the intuition suppose that all processes lose.
Consider the last write, say $w$, to $A$, and let $p$ be the process that executes $w$.
Process $p$ cannot lose on $A$ because if it did,
then in $p$'s last scan there is some process, $q \neq p$, that occupies 2 positions on $A$,
and that process cannot return \lose unless some process writes to $A$ after $w$. 
Similarly,  $p$ cannot signature-lose on $B$ because, again, that would imply a write to
$A$ after $w$.
So suppose $p$ process-loses on $B$.
Then we show that there is some other process, say $q$,
that has the same signature as $p$ and is competing with $p$, and $q$ cannot process-lose on $B$.
Process $q$ also cannot signature-lose on $B$ or lose on $A$ without a write to $A$ happening after $w$.

Lemma~\ref{lem:main} below states that if $k$ processes call \Competition{}, then at most $\floor{(2k+1)/3}$ of them win. Consider the intervals in an execution between the final scans of processes
that return \win (achieve a clean-sweep of $A$).
If $\ell$ processes return \win, there are $\ell -1$ such disjoint intervals.
We associate each such interval $I$, with a losing process 
as follows.
\begin{enumerate}
	\item If $I$ contains the last write by a process $p$ that loses on $A$, 
then
associate $I$ with $p$. 
\item If $I$ contains the last write by a process $q$ that signature-loses on $B$,
associate $I$ with $q$. 
\item If $I$ is not associated with a losing process via either (1) or (2),
we will associate $I$ with a losing process as follows.
\end{enumerate}
We will prove that there is a sub-interval $I'$ of $I$ and there are either two or three processes that,
during $I'$,
move from $A$ to $B$ and finish competing on $B$ using some signature, say $s$,
while the signature of $A$ remains $s$ throughout $I'$.
Now we focus on the execution during $I'$.
Since $B$ has three components, after any clean-sweep on $B$,
a subsequent clean-sweep on $B$ requires two processes to over-write the previous clean-sweep.
These over-writers must have signature $s$,
because, otherwise, an over-writer has a signature different from that of $A$ and
would signature-lose on $B$,
implying that $I$ has an associated losing process via (2).
If there are two processes with signature $s$ then $I'$ can have at most one clean-sweep,
and if there are three processes then $I'$ can have at most two clean-sweeps.
Therefore, at least one of the two or three processes competing on $B$ with signature $s$ cannot return \win,
and $I$ is associated with one such process. 
Notice, however, that this process could withhold its last write in order to be assigned to a later interval via~(2).

Therefore, using these three rules of association, we assign at least one losing process to every interval, and
no process is assigned to more than two of these intervals.
Thus there are at least $(\ell-1)/2$ processes that cannot return \win.

\subsection{Notation and Terminology}
Throughout the remainder of the section we consider a fixed execution $E$.
A \emph{losing scan} is a scan by a process such that this process will return \lose in its next step,
without doing any further shared memory operation.
A \emph{winning scan} is a scan by a process such that this process will return \win in its next step,
without doing any further shared memory operation.
For each winning scan there exists a last write by the process that performs this scan.
We call this write a \emph{winning write}.
Let $s_1,s_2,\dots,s_\kappa$ be the sequence of winning scans in $E$
and let  $q_1,q_2,\dots,q_\kappa$ denote the corresponding sequence of processes that performed these scans.
Observe that for all $i$, $1 \leq i \leq \kappa$,  $s_i$ is preceded by a winning write $w_i$ performed by $q_i$.
Furthermore, $s_i$ must happen before $w_{i+1}$
because at $s_i$ all components in $A$ contain $q_i$'s $\id$
however, at $w_{i+1}$, $q_{i+1}$ has written its own $\id$ everywhere in $A$.
Hence winning scans and winning writes strictly interleave.
That is, the order of winning scans and writes in $E$ is $w_1,s_1,w_2,s_2,\dots,w_\kappa,s_\kappa$.

Suppose $E = \mathrm{op}_1, \mathrm{op}_2 \ldots $,
we denote the contiguous subsequence of $E$
starting at $\mathrm{op}_i $ and ending at the operation immediately before $\mathrm{op}_j$
by $E[\mathrm{op}_i: \mathrm{op}_j)$.
A \emph{sifting interval} is a subsequence of an execution
that starts at some winning scan and ends at the operation immediately before the next winning write.
Observe that all sifting intervals are disjoint.
Also because there has been a preceding winning scan, no component of $A$ contains $\bot$ in any sifting interval.
Note that since $E$ contains $\kappa$ winning scans it has $\kappa-1$ disjoint sifting intervals.

A \emph{signature} is an ordered triple of identifiers.
A signature $(p_0,p_1,p_2)$ is \emph{full} if for any $i,j \in \{0,1,2\}, i \neq j$ implies $p_i \neq p_j$.

The following lemmas concern properties of executions.
Terms such as
before, after, next, previous, precedes, and follows are all with respect to the order of operations in execution $E$.

A local variable $x$ in the algorithm is denoted by $x_p$ when it is used in the method call invoked by process $p$.

\subsection{Proof of Correctness}

Lemmas~\ref{lem:scans} through \ref{lem:noOld-sig} provide us with some properties of the algorithm that are used in Lemma~\ref{lem:lose}, to prove that there is no execution in which all processes return \lose.

\begin{lemma}\label{lem:scans}
Suppose that process $p$ executes a scan, say $s$, at Line~\ref{line:Kscan}, and in this scan$(A,B)$, $A=\sigma$.
If $s$ is not a losing scan then, at the most recent scan$(A)$ executed in Line~\ref{line:Cscan}, by $p$, preceding $s$, $A = \sigma$.
\end{lemma}

\begin{proof}
Let $\hat{s}$ be the most recent scan$(A)$ executed in Line~\ref{line:Cscan}, by $p$, preceding $s$.
By way of contradiction suppose that at $\hat{s}$, $A = \sigma'$, where $\sigma' \neq \sigma$.
Then, at $s$, by Line~\ref{line:invokeK}, $\sig_p = \sigma'$ and by Line~\ref{line:Kscan}, $\widehat{a}_p = \sigma$. Hence at $s$, $\widehat{a}_p \neq \sig_p$. Therefore, by Lines~\ref{line:A-changed} and \ref{line:second-die}, $s$ is a losing scan which is a contradiction.
\end{proof}

\begin{lemma}\label{obs:diffPos}
Let $s$ be any scan by process $p$ and $w$ be $p$'s next write.  If, at $w$, $p$ writes to $A[j]$, then at $s$, $A[j] \neq p$ and $A[(j-1) \bmod 3] = p$.
\end{lemma}
\begin{proof}
Let $\hat{s}$ be the last scan$(A)$ executed in Line~\ref{line:Cscan} by $p$ preceding $w$.
Since $w$ is to $A[j]$, by Line~\ref{line:pos}, at $\hat{s}$, $A[j] \neq p$ and $A[(j-1) \bmod 3] = p$.
Since $p$ performs a write after $s$, $s$ is not a losing scan.
By Lemma~\ref{lem:scans}, the signature of $A$ at $s$ and $\hat{s}$ is equal.
Therefore at $s$, $A[j] \neq p$ and $A[(j-1) \bmod 3] = p$.
\end{proof}

\begin{lemma}\label{lem:change-location}
Suppose at scan $s$, $A=(p_0,p_1,p_2)$ is a full signature.
For any $i \in \{0,1,2\}$, if $p_i$ writes to $A$ after $s$, then its first write into $A$ after $s$ is not to $A[i]$.
\end{lemma}
\begin{proof}
Let $w_i$ be the first write by $p_i$ to $A$ after $s$.
Let $s_i$ be the scan by $p_i$ preceding $w_i$.
If $s_i$ happens before $s$,
then
there is no write to $A$ by $p_i$ in the execution $E[s_i : s)$.
At $s$, $A[i]=p_i$, hence at $s_i$, $A[i]=p_i$.
Suppose $s_i$ happens after $s$.
At $s$, $A[i]$ is the only location that contains $p_i$,
and there is no write to $A$ by $p_i$ in the execution $E[s : s_i)$
and $s_i$ is not a losing scan.
Therefore, at $s_i$, $A[i]=p_i$.
In either case, by Lemma~\ref{obs:diffPos}, $w_i$ is a write to $A[\pos]$ where $\pos \neq i$.
\end{proof}

\begin{lemma}\label{lem:first-write}
Suppose at scan $s$, $A=(p_0,p_1,p_2)$ is a full signature.
Let $w$ be the first write to $A$ after $s$. Then $w$ changes the signature of $A$.
\end{lemma}

\begin{proof}
Let $q$ be the process executing $w$.
If $q \notin\{p_0,p_1,p_2\}$, then since $q$ writes its own id,
it changes the signature of $A$.
If $q = p_i \in \{p_0,p_1,p_2\}$, then by Lemma~\ref{lem:change-location},
$q$ writes to a location different from $A[i]$.
Hence $w$ changes the signature of $A$.
\end{proof}

\begin{lemma}\label{lem:allParticipate}
Suppose at scan $s_1$, $A=(p_0,p_1,p_2)$ is a full signature.
Let $w$ be the first write to $A$ after $s_1$.
Let $s_2$ be any scan after $w$
such that at $s_2$, $A=(p_0,p_1,p_2)$. Then, for some $\ell \in \{0,1,2\}$, $p_{\ell}$ calls  \Knockout{$\sigma$}, where $\sigma \neq (p_0,p_1,p_2)$ and returns \False in the execution $E[w:s_2)$.
\end{lemma}

\begin{proof}
Suppose that $w$ is a write to component $A[i]$.
Since $A[i]=p_i$ at $s_2$, the last write to $A[i]$ in $E[w:s_2)$, is by $p_i$.

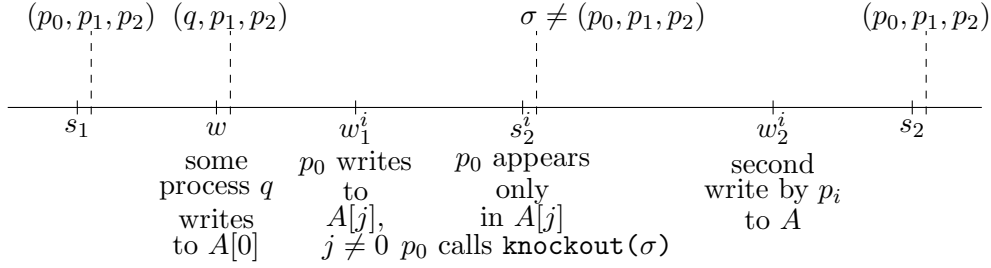
\begin{figure}[ht]
	\centering

\begin{tikzpicture}[scale=1.85]
\draw (-0.5,0) -- (6.5,0);

\draw (0,-0.05) -- (0,0.05);
\node at (0,-0.15) {$s_1$};

\draw (0.1,-0.05)[dashed] -- (0.1,0.55);
\node at (0.1,0.65) {$(p_0,p_1,p_2)$};

\draw (1,-0.05) -- (1,0.05);
\node at (1,-0.15) {$w$};
\node at (1,-0.4) {some};
\node at (1,-0.6) {process $q$};
\node at (1,-0.8) {writes};
\node at (1,-1) {to $A[0]$};

\draw (1.1,-0.05)[dashed] -- (1.1,0.55);
\node at (1.1,0.65) {$(q,p_1,p_2)$};

\draw (2,-0.05) -- (2,0.05);
\node at (2,-0.15) {$w^i_1$};
\node at (2,-0.4) {$p_0$ writes};
\node at (2,-0.6) {to};
\node at (2,-0.8) {$A[j],$};
\node at (2,-1) {$j \neq 0$};

\draw (3.2,-0.05) -- (3.2,0.05);
\node at (3.2,-0.15) {$s^i_2$};
\node at (3.2,-0.4) {$p_0$ appears};
\node at (3.2,-0.6) {only};
\node at (3.2,-0.8) {in $A[j]$};
\node at (3.3,-1) {$p_0$ calls \Knockout{$\sigma$}};

\draw (3.3,-0.05)[dashed] -- (3.3,0.55);
\node at (3.85,0.65) {$\sigma \neq (p_0,p_1,p_2)$};

\draw (5,-0.05) -- (5,0.05);
\node at (5,-0.15) {$w^i_2$};
\node at (5,-0.4) {second};
\node at (5,-0.6) {write by $p_i$};
\node at (5,-0.8) {to $A$};

\draw (6,-0.05) -- (6,0.05);
\node at (6,-0.15) {$s_2$};

\draw (6.1,-0.05)[dashed] -- (6.1,0.55);
\node at (6.1,0.65) {$(p_0,p_1,p_2)$};

\end{tikzpicture}

	\label{figure:a}
	\caption{Illustration of the order of the operations, where $i = 0$}
\label{figure:OpsA}
\end{figure}

By Lemma~\ref{lem:change-location}, the first write by $p_i$ to $A$ in $E[w:s_2)$, say $w^i_1$, is to $A[j]$ where $j \neq i$.
Thus $p_i$ must perform at least two writes to $A$ in the interval $E[w : s_2)$.
Let $s^i_2$ be the scan by $p_i$ following $w^i_1$ in $E[w : s_2)$, and
$\sigma$ be the signature of $A$ at $s^i_2$.

Since $w^i_1$ is to $A[j]$ and $w$ is to $A[i]$, $w \neq w^i_1$, implying $w$ is not executed by $p_i$.
Immediately after $w$, no location in $A$ contains $p_i$.
Because $p_i$ writes only once in $E[w : s^i_2)$, at $s^i_2$, $p_i$ can appear only in $A[j]$. Since $s^i_2$ is not a losing scan, $p_i$ must still be in $A[j]$ at $s^i_2$, and $\sigma$ be full.
This implies $p_i$ calls \Knockout{$\sigma$} after $s^i_2$.
Furthermore, $\sigma \neq (p_0,p_1,p_2)$.
Finally, because $p_i$ writes to $A$ after $s^i_2$,
$p_i$ must return \False from this \Knockout call.
\end{proof}

Recall that in any execution, if a process $p$ performs a scan$(A)$ in which $A = \sigma$ is a full signature containing $p$, it invokes \Knockout{$\sigma$}. During this \Knockout{$\sigma$} call, $p$ tries to write $(p, \sigma)$ to all components of $B$. Let $w$ be any write of this \Knockout{$\sigma$} call. In the following lemma, we prove that if at some scan after $w$, say $s$, $A = \sigma$ and $B[i] = (p, \sigma)$, then the signature of $A$ is $\sigma$ in the entire execution between $w$ and $s$. In other words, during $E[w : s)$, the signature of $A$ cannot change from $\sigma$ to $\sigma' \neq \sigma$ and change back to $\sigma$ again while $p$ is performing one single \Knockout{$\sigma$}.

\begin{lemma}\label{lem:noOld-sig}
Suppose at scan $s$, $A= \sigma$ is a full signature, and
there is an $i \in \{0,1,2\}$ and a process $p$ such that $B[i] = (p,\sigma)$.
Let $w_i$ be the last write to $B[i]$ that precedes $s$.
Then, there is no write to $A$ in $E[w_i:s)$.
\end{lemma}

\begin{proof}
By way of contradiction, let $w$ be the first write to $A$ in $E[w_i:s)$.
Let $s_i$ be the last scan by $p$ preceding $w_i$.
Since $w_i$ has value $(p, \sigma)$, at $w_i$, $\sig_p = \sigma$. Therefore, at the last scan executed in Line~\ref{line:Cscan} preceding $w_i$, $A=\sigma$.

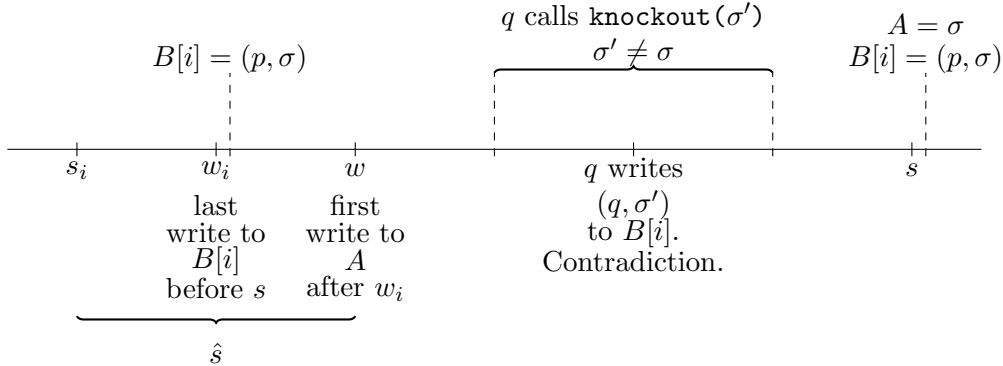
\begin{figure}[hb]
	\centering

\begin{tikzpicture}[scale=1.85]
\draw (-0.5,0) -- (6.5,0);

\draw (0,-0.05) -- (0,0.05);
\node at (0,-0.15) {$s_i$};

\draw (1,-0.05) -- (1,0.05);
\node at (1,-0.15) {$w_i$};
\node at (1,-0.4) {last};
\node at (1,-0.6) {write to};
\node at (1,-0.8) {$B[i]$};
\node at (1,-1) {before $s$};

\draw (1.1,-0.05)[dashed] -- (1.1,0.55);
\node at (1.1,0.65) {$B[i] = (p, \sigma)$};

\draw (2,-0.05) -- (2,0.05);
\node at (2,-0.15) {$w$};
\node at (2,-0.4) {first};
\node at (2,-0.6) {write to};
\node at (2,-0.8) {$A$};
\node at (2,-1) {after $w_i$};

\draw[thick,decoration={brace,mirror},decorate]
  (0,-1.2) -- node[below=6pt ] {$\hat{s}$} (2,-1.2);

\draw (3,-0.05)[dashed] -- (3,0.55);
\draw[thick,decoration={brace},decorate]
  (3,0.55) -- node[above=12pt] {$q$ calls \Knockout{$\sigma'$}} (5,0.55);
\node at (4,0.7) {$\sigma' \neq \sigma$};
\draw (5,-0.05)[dashed] -- (5,0.55);

\draw (4,-0.05) -- (4,0.05);
\node at (4,-0.15) {$q$ writes};
\node at (4,-0.4) {$(q, \sigma')$};
\node at (4,-0.6) {to $B[i]$.};
\node at (4,-0.8) {Contradiction.};

\draw (6,-0.05) -- (6,0.05);
\node at (6,-0.15) {$s$};

\draw (6.1,-0.05)[dashed] -- (6.1,0.55);
\node at (6.1,0.65) {$B[i] = (p, \sigma)$};
\node at (6.1,0.9) {$A = \sigma$};

\end{tikzpicture}

	\label{figure:3}
	\caption{Illustration of the order of the operations if $w_i$ precedes $\hat{w}$}
\end{figure}

Hence, by Lemma~\ref{lem:scans}, at $s_i$, $A$ must have signature $\sigma$.
 Let $\hat{s}$ be the last scan before $w$ in which the signature of $A$ is $\sigma$.
Since $s_i$ precedes $w$, $\hat{s}$ exists.
Then $w$ is the first write to $A$ following $\hat{s}$.
 By Lemma~\ref{lem:allParticipate}, there is a process $q$ that executes a complete \Knockout{$\sigma'$} in $E[w : s)$, where $\sigma \neq \sigma'$, and returns \False.
Hence, in $E[w_i : s)$, $q$ over-writes every component in $B$ with $(q, \sigma')$.
This contradicts that $w_i$ is the last write to $B[i]$ preceding $s$.
\end{proof}

\begin{lemma}\label{lem:lose}
There is no execution in which all processes return \lose.
\end{lemma}
\begin{proof}
  By way of contradiction, assume that there is an execution in which all processes return \lose.
    Let $u$ be the process that performs the last write to $A$,
  let $w_u^{A}$ be that write,
  and let $\sigma$ be the signature of $A$ after $w_u^A$.
  Let $s_u$ be the last scan by $u$.
  Then $u$ returns \lose in Line~\ref{line:first-die} or \ref{line:second-die}.

First consider the case in which $u$ returns \lose in Line~\ref{line:first-die}.
At $s_u$, \num{u}{a_u} is not equal to $0$ because the last write to $A$ is performed by $u$ and $s_u$ happens after $w_u^A$.
Therefore, by the if-condition of Line~\ref{line:second-if},
\num{u}{a_u}$ = 1$ and there is a process $y$ such that in $s_u$,
\num{y}{a_u} $=2$.
Let $w_y^A$ be the last write by $y$ to $A$.
Since $u$ performs the last write to $A$, $w_y^A$ precedes $w_u^A$.
  Because no process writes $y$ to $A$ after $w_y^A$ and no process writes to $A$ after $w_u^A$ and, later, at $s_u$, \num{y}{a_u} $=2$,
  it follows that \num{y}{A} $\geq 2$ for the entire execution after $w_y^A$.
  Therefore any scan
	by $y$ after $w_y^A$ must satisfy \num{y}{a_y}$\geq 2$.
  This implies $y$ cannot return \lose, contradicting the assumption.

Next consider the case in which $u$ returns \lose in Line~\ref{line:second-die}.
This implies $u$ calls \Knockout{$\sigma$} after $w_u^{A}$
from which it returns \True in Line~\ref{line:first-true} or in Line~\ref{line:second-true}.
But $u$ cannot return \True in Line~\ref{line:first-true} because the value of array $A$
remains $\sigma$ after $w_u^A$.
Therefore $u$ returns \True in Line~\ref{line:second-true}.

Let $S = \{ (q, i, w_q) ~|~ B[i] = (q, \sigma) \text{ at some scan after $w_u^A$ and $w_q$ is the last write by $q$}$ $\text{ to $B[i]$ before this scan}\}$ and let $Q=\{q ~|~ (q, i, w_q) \in S\}$.
Because $u$ returns \True in Line~\ref{line:second-true},
$S$ is not empty.
By Lemma~\ref{lem:noOld-sig}, for each $(q, i, w_q) \in S$, $w_u^A$ precedes $w_q$.
This implies that for each $q \in Q$, $q$ performs a write (i.e. $w_q$) to $B$
after $w_u^A$ and, by assumption, some time later, does a losing scan.

\begin{figure}[h]
	\centering

\begin{tikzpicture}[scale=1.85]
\draw (0.5,0) -- (7.5,0);

\draw (1,-0.05) -- (1,0.05);
\node at (1,-0.15) {$w^A_u$};
\node at (1,-0.4) {last};
\node at (1,-0.6) {write to};
\node at (1,-0.8) {$A$};
\node at (1,-1) {in the execution};

\draw (3,-0.05) -- (3,0.05);
\node at (3,-0.15) {$w_q$};
\node at (3,-0.4) {last};
\node at (3,-0.6) {write to};
\node at (3,-0.8) {$B[i]$};
\node at (3,-1) {before $s_q$};

\draw (3.1,-0.05)[dashed] -- (3.1,0.55);
\node at (3.1,0.65) {$B[i] = (q, \sigma)$};

\draw (7,-0.05) -- (7,0.05);
\node at (7,-0.15) {a scan $s_q$};
\node at (7,-0.35) {after $w^A_u$};

\draw (7.1,-0.05)[dashed] -- (7.1,0.55);
\node at (7.1,0.65) {$B[i] = (q, \sigma)$};
\node at (7.1,0.9) {$A = \sigma$};

\end{tikzpicture}

	\label{figure:5}
	\caption{Illustration of the order of the operations if $w_i$ precedes $\hat{w}$}
\end{figure}
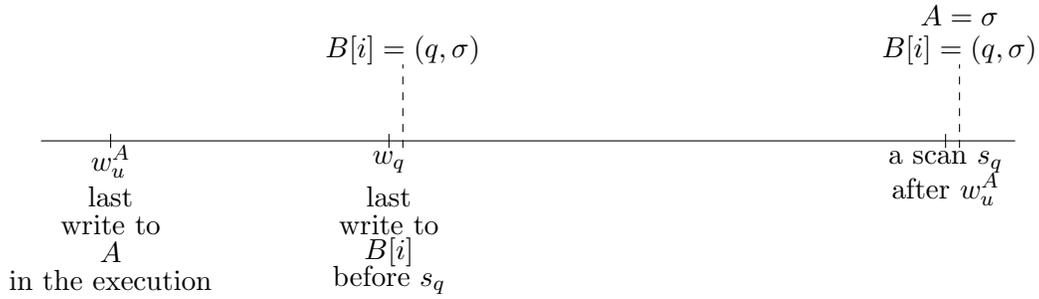

For each $q \in Q$,
$q$ cannot return \lose at Line~\ref{line:first-die} because this would imply $q$ writes to $A$ after $w_u^A$.
Therefore $q$ returns \lose at Line~\ref{line:second-die} implying that $q$ returns \True at Line~\ref{line:first-true} or Line~\ref{line:second-true}.
It does not return \True in Line~\ref{line:first-true} because $\sig_q = \sigma$
and the value of $A$ remains $\sigma$ after $w_u^A$.
Therefore for each $q \in Q$, $q$ returns \True in Line~\ref{line:second-true},
following a losing scan that is after $w_u^A$.
Let $z$ be the last process in $Q$ to do its losing scan, $s_z$.
At $s_z$ two components in $B$ contain $(z', \sigma)$, where $z' \neq z$.
Hence, $z' \in Q$.
Thus, between the last write by $z'$ (after $w_u^A$) and the last scan by $z'$,
these two components in $B$ contain $(z', \sigma)$.
So at $z'$'s last scan at least two components in $B$ contain $(z', \sigma)$. Therefore $z'$ cannot return \True in Line~\ref{line:second-true},
contradicting the assumption that the last scan of $z'$ is a losing scan.
\end{proof}
Lemma~\ref{lem:late-scan} through Lemma~\ref{lem:main2} provide us with additional properties of the algorithm
that are combined to prove, in Lemma~\ref{lem:math}, that  when two or more processes invoke \Competition{},
at most a constant fraction of them can return \win.

\begin{lemma}\label{lem:late-scan}
Suppose an execution between a write and the next scan by the same process, say $p$, contains a winning write.
Then the scan by $p$ is a losing scan.
\end{lemma}
\begin{proof}
At the winning write all components in $A$ contain the $\id$ of the process that performs this winning write.
In the sub-execution from the winning write to the scan by $p$ there is no write by $p$.
Since only $p$ writes its $\id$, at $p$'s scan, \num{p}{a_p} $=0$.
Hence $p$ returns \lose after this scan.
\end{proof}

\begin{observation}\label{obs:two-writes}
Every sifting interval contains at least two writes to $A$.
\end{observation}
\begin{proof}
Consider the sifting interval $I=E[s^{\ast}:w^{\ast})$. Let $p^{\ast}$ be the process that performs $w^{\ast}$.
Since at the winning write $w^{\ast}$, all components in $A$ contain $p^{\ast}$,
$p^{\ast}$ must have performed at least two writes to $A$ before $w^{\ast}$,
and these two writes must be after the previous winning scan, which is $s^{\ast}$.
\end{proof}

A sifting interval that does not contain a write to $A$ by a process whose next scan is a losing scan
is called a \emph{slow sifting interval}.

\begin{lemma}\label{lem:help}
For any slow sifting interval $I$,
there exists a signature $\sigma=(q_0,q_1,q_2)$
and a set $Z \subseteq \{0, 1, 2\}$ satisfying: $|Z|=2$ and for each $z \in Z$ during $I$,
 $q_z$ performs a write
and then a scan in \Competition{}
and then invokes \Knockout{$\sigma$}
and becomes poised to write $(q_z, \sigma)$ to $B$.
Furthermore, there is no write to $A$ between these two scans.
\end{lemma}

\begin{proof}
Let $I$ be $E[s^{\ast}:w^{\ast})$ and Let $p^{\ast}$ be the process that performs $s^{\ast}$.
Suppose that $w_1^A,w_2^A,\dots,w_{\ell}^A$ is the sequence of all writes to $A$ during $I$.
By Observation~\ref{obs:two-writes}, $\ell \geq 2$.
For each $i$, $1 \leq i \leq \ell$, let $s_i$ denote the next scan by the process that executes $w_i^A$.
Each $s_i$ is at Line~\ref{line:Cscan} following $w_i$, at Line~\ref{line:Awrite}
of \Competition{}.
Let $S$ denote the set of all these scans.
Let $\hat{I}$ denote the execution $E[w_2^A : w^{\ast})$.

By Lemma~\ref{lem:late-scan}, if $s_i$ happens after $w^{\ast}$
then $s_i$ is a losing scan and hence $E[s^{\ast}:w^{\ast})$ is not a slow sifting interval.
Therefore for all $i$, $1 \leq i \leq \ell$, $s_i$ occurs in $E[w_i^A : w^{\ast})$.
Let $q$ be the process that performs $s_1$.
At $s^{\ast}$, \num{p^{\ast}}{A} $=3$.
Because only one write happens to $A$ during $E[s^{\ast}: w_2^A)$,
$q$ would return \lose at Line~\ref{line:first-die} if $s_1$ precedes $w_2^A$
implying $E[s^{\ast}:w^{\ast})$ is  not a slow sifting interval.
Hence for all $i$, $1 \leq i \leq \ell$, $s_i$ must happen in $\hat{I}$.
Execution $\hat{I}$ consists of the $\ell -1$ disjoint sub-executions
$E[w_2^A:w_3^A), E[w_3^A:w_4^A), \ldots, E[w_{\ell}^A: w^{\ast}=w_{\ell+1}^A)$.
Since $\ell$ scans happen in these $\ell-1$ executions, by the pigeonhole principal,
there is a $j$, $2 \leq j \leq \ell$ such that (at least) two scans in $S$, say $s'$ and $s'' $  occur in $E[w_j^A: w_{j+1}^A)$.
Because no process performs two scans in \Competition{} without writing to $A$ in between,
$s'$ and $s''$ are performed by two distinct processes say $q_z$ and $q_{z'}$.
Because  $E[s^{\ast}:w^{\ast})$ is a  slow sifting interval,
neither $q_z$ nor $q_{z'}$ return \lose at Line~\ref{line:first-die}.
Since no write happens to $A$ during $E[w_j^A,w_{j+1}^A)$,
the scans by $q_z$ and $q_{z'}$ in \Competition{} return the same signature for $A$, say, $\sigma$ where $\sigma$ contains $q_z$ and $q_{z'}$.
Therefore $q_z$ and $q_{z'}$ both invoke  \Knockout{$\sigma$}.
\end{proof}

\begin{lemma}\label{lem:allParticipateb}
Suppose at scan $s_1$, $A=(p_0,p_1,p_2)$ is a full signature.
Let $w$ be the first write to $A$ after $s_1$.
Let $s_2$ be any scan after $w$ such that at $s_2$, $A=(p_0,p_1,p_2)$. Then, for all $\ell \in \{0,1,2\}$, $p_{\ell}$ performs at least two writes to $A$ in the execution $E[w : s_2)$.
\end{lemma}

\begin{proof}
In order to prove the lemma we show that for each $\ell \in \{0,1,2\}$, in the execution $E[w : s_2)$:

\begin{enumerate}[a)]
	\item if $A[\ell]$ is written, then in this execution, $p_{\ell}$ writes to $A$ at least twice;
	\item $A[\ell]$ is written.
\end{enumerate}

First we prove $(a)$. Let $w_{\ell}$ be the last write to $A[\ell]$ during $E[w:s_2)$.
Since $A[\ell]=p_{\ell}$ at $s_2$, $w_{\ell}$ is executed by $p_{\ell}$.
By Lemma~\ref{lem:first-write}, $p_{\ell}$'s first write during $E[w:s_2)$ is not to $A[\ell]$. Hence, $p_{\ell}$ executes at least two writes during $E[w:s_2)$, proving $(a)$.

We now prove $(b)$. Suppose that $w$ is a write to component $A[i]$. By $(a)$, $p_i$ writes to $A$ at least twice during $E[w:s_2)$.
Let $w^A_i$ be $p_i$'s first write to $A$ during $E[w:s_2)$. By Lemma~\ref{lem:change-location}, $w^A_i$ is to $A[j]$ where $j \neq i$.

By way of contradiction assume there is a $k \in \{0,1,2\}$ such that, $A[k]$ is not written in $E[w : s_2)$.
In particular, since $A[i]$ and $A[j]$ are written (by writes $w$ and $w^A_i$, respectively) in $E[w : s_2)$, we have:

 \begin{align}
	\text{$k \notin \{i,j\}$}.
 \end{align}

Let $w_{i'}$ and $w_{j'}$ be the last writes to $A[i]$ and $A[j]$, respectively, during $E[w:s_2)$.
Since $A[i]=p_i$ and $A[j] = p_j$ at $s_2$, $w_{i'}$ is executed by $p_i$ and $w_{j'}$ by $p_{j}$.
Let $s_{i'}$ and $s_{j'}$ be the scans by $p_i$, respectively $p_j$, preceding $w_{i'}$, respectively $w_{j'}$.
By $(a)$, both processes execute at least two writes during $E[w:s_2)$, and thus $s_{i'}$ and $s_{j'}$ are both also in $E[w:s_2)$. From Lemma~\ref{obs:diffPos} we conclude that $A[ (i-1) \bmod 3]=p_i$ at $s_{i'}$ and $A[(j-1) \bmod 3]=p_j$ at $s_{j'}$. Since no process writes to $A[k]$ in $E[w : s_2)$, $A[k]=p_k$ throughout $E[w:s_2)$. Hence, we have:

\begin{align}
	\text{$k\neq (i-1) \bmod 3$ and $k\neq (j-1) \bmod 3$}.
 \end{align}

Combining conditions $(4.1)$ and $(4.2)$ contradicts that $k$ is in $\{0,1,2\}$.
\end{proof}

\begin{lemma}\label{lem:noRepeat}	
Let $s$ be a scan$(A)$ from Line~\ref{line:Cscan} by $p$ immediately before $p$ invokes \Knockout{$\sigma$} and
$s'$ be any scan$(A, B)$ by $p$ within this invocation.
Let $w$ be the first write to $A$ after $s$.
If  $w$ precedes $s'$, then $s'$ is a losing scan.
\end{lemma}

\begin{proof}
Since $p$ invokes \Knockout{$\sigma$}, $\sigma$ is a full signature, $\num{p}{\sigma} = 1$ and $A = \sigma$ at $s$.
By way of contradiction suppose $s'$ is not a losing scan.
Hence, by Lemma~\ref{lem:scans}, at $s'$, the signature of $A$ is $\sigma$.
Therefore, by Lemma~\ref{lem:allParticipateb}, $p$  writes to $A$ in the execution $E[w : s')$.
This is a contradiction because $p$ is performing \Knockout{$\sigma$} in this entire execution and there are no writes to $A$ during the \Knockout method call.
\end{proof}

\begin{lemma}\label{lem:main}
For every slow sifting interval $I$,
there is a  process $p$ that performs a write during $I$
and either the first or the second scan by $p$
following this write is a losing \scan.
\end{lemma}

\begin{proof}
Let $I = E[s^{\ast}:w^{\ast})$ be a slow sifting interval.
By Lemma~\ref{lem:help}, there exists a full signature $\sigma = (q_0,q_1,q_2)$, a set $Q \subseteq \{q_0,q_1,q_2\}$, satisfying $|Q| = 2$ and for each $q \in Q$
during $I$:
\begin{enumerate}[1)]
	\item $q$ performs a write to $A$ and a scan in \Competition{}
and calls \Knockout{$\sigma$} and becomes poised,
at Line~\ref{line:Bwrite}, to write $(q, \sigma)$ to $B[0]$; and
\item there is no write to $A$ between these scans.
\end{enumerate}
Let $Q' \subseteq \{q_0,q_1,q_2\}$ be the set of all processes satisfying $(1)$ and $(2)$. Therefore $2 \leq |Q'| \leq 3$.
Let $\hat{s}$ be the earliest of these scans (by processes in $Q'$ immediately before calling \Knockout{$\sigma$}).
At $\hat{s}$, the signature in $A$ is full and at $w^{\ast}$ the same id is in all locations of $A$.
Therefore, $w^{\ast}$ is  the second  or later write after $\hat{s}$.
Hence, there is at least one write to $A$ in $E[\hat{s}: w^{\ast})$.
Let $w$ be the first write to $A$ in $E[\hat{s}: w^{\ast})$.

Suppose there is $q \in Q'$, such that $q$ performs a scan, say $s$, in Line~\ref{line:Kscan}
of its current call to \Knockout after $w$. Then by Lemma~\ref{lem:noRepeat}, $s$ is a losing scan.
Since $q$ writes at least once in $E[s^{\ast}:w^{\ast})$
and at most once after $w$,
it follows that $q$ performs its last or second last write during $E[s^{\ast}:w^{\ast})$,
and so $s$ is either $q$'s first or second scan following this write, and the lemma holds.

Otherwise, all processes in $Q'$ execute at least one write and perform their last scan of their current call to \Knockout before $w$.
We partition this case into three subcases.

Case 1: There is $q \in Q'$ such that $q$ calls \Knockout{$\sigma$} and returns \True
 (Line~\ref{line:first-true} or \ref{line:second-true}).
Then $q$'s last scan before returning \True
is a losing scan, and the lemma follows.

Case 2: For each process $q \in Q'$, $q$'s current \Knockout call returns \False  and
there is a process $p \notin Q'$
that performs a write $w_p$ to $B$ with value $(p,\sigma')$ in the execution $E[\hat{s}:w)$ where $\sigma' \neq \sigma$.
When $p$ did its scan in \Competition{} just before invoking \Knockout{$\sigma'$},
the signature of $A$ was $\sigma'$.
At $w_p$, the signature of $A$ is $\sigma \neq \sigma'$, so there is a write to $A$ between this scan by $p$
and $w_p$.
Hence, by Lemma~\ref{lem:noRepeat}, $p$'s next scan after $w_p$ is a losing scan, and again the lemma follows. 		

Case 3: For each process $q \in Q'$, $q$'s current \Knockout call returns \False  and
there is no write to $B$ in $E[\hat{s}:w)$ that contains a signature different from $\sigma$. We show that this case is impossible.
Let $S$ be the set of last scans of \Knockout calls by processes in $Q'$.
Let $s''$ be the last scan and $s'$ be the second last scan in set $S$. Let $q'$ and $q''$ be the processes performing $s'$ and $s''$ respectively.
Since $q'$ returns \False, all three components in $B$ contain $(q',\sigma)$ at $s'$.
After $s'$, there can be at most one write to $B$ by $q''$.
Because $q''$'s next scan after such a write would be a losing scan,
contradicting that $q''$ returns \False.
\end{proof}

\begin{lemma}\label{lem:main2}
For every sifting interval, there is a process $p$ and a write $w$ by $p$ satisfying: either
the first operation by $p$ or the third operation by $p$ that follows $w$ is a losing scan.
\end{lemma}

\begin{proof}
For any sifting interval that is not slow, the lemma holds by definition.
For any slow sifting interval,
the lemma follows from Lemma~\ref{lem:main},
because each process alternates between writes and scans.
\end{proof}

\begin{lemma}\label{lem:math}
	If $k$ processes invoke the \Competition{} method, then at most $\floor {\frac{2k+1}{3}} $  processes  return \win.
\end{lemma}
\begin{proof}
	If $k'$ processes return \win, then by definition, there are $k'-1$ sifting intervals.
	By Lemma~\ref{lem:main}, for each sifting interval there is a process that performs its last or second last write and it cannot return \win.
	Hence there are at least $\ceil {\frac{k'-1}{2}} $ processes which have invoked \Competition{} and cannot return \win.
	Since $\ceil {\frac{k'-1}{2}} + k' \leq k$, $k'$ is at most $\floor {\frac{2k+1}{3} }$.
\end{proof}
\begin{lemma}\label{lem:obs-free}
The sifter implementation in Figure \ref{fig:Tas-alg} is obstruction-free where each process terminates in $O(1)$ solo steps.
\end{lemma}

\begin{proof}
 Suppose a process, $p$, begins a solo run while it is executing \Knockout.
 If it returns \True in either Line~\ref{line:first-true} or Line~\ref{line:second-true},
  then it terminates due to Line~\ref{line:second-die}.
  Otherwise in each iteration of the while loop, it writes a new location in $B$.
  Therefore after three iterations, all locations in $B$ contain $(p,\sig_p)$, and
  $p$ returns \False in Line~\ref{line:Bwin}.
When $p$ executes \Knockout during its solo run, the value of $A$ is equal to $\sig_p$ because otherwise $p$ returns \True from its \Knockout call.
In $\sig_p$, exactly one location in $A$ contains $p$ and no other process writes to $A$  after it returns from
its \Knockout call.
Hence $p$ writes two more times to $A$ and, by Line~\ref{line:Awin} returns \win.

Suppose $p$ starts its solo run in a \Competition{} call.
After at most one write it performs a scan.
Then it either returns \win due to Line~\ref{line:Awin-condition} or returns \lose in Line~\ref{line:first-die},
or it invokes a \Knockout call.
If it calls \Knockout, then by the argument above it terminates.
\end{proof}
Notice that each component of the snapshot object used in our sifter implementation in Figure~\ref{fig:Tas-alg} holds at most 4 identifiers,
so it is a $(4 \cdot \log {n})$-bounded 6-component snapshot object.
Combining this with Lemma~\ref{lem:lose}, Lemma~\ref{lem:math} and  Lemma~\ref{lem:obs-free} yields Theorem~\ref{thm:sifter}.

\section{Obstruction-Free Snapshot from Registers}\label{sec:snapshot}
This section establishes Theorem \ref{thm:scan}.
That is, we present an obstruction-free implementation of a $B$-bounded $M$-component snapshot object
from $M+1$ registers of size $\Theta(B + \log{n})$.

Our implementation uses an array $A[1\dots M]$ of shared registers and a register $S$.
Each array entry $A[i]$ stores a triple $(w_i,p_i,b_i)$, where $w_i\in D$ represents the $i$-th entry in the vector $V$ of the snapshot object, $p_i$ is a process ID or $\bot$ which identifies the last process that wrote to $A[i]$, and $b_i\in\{0,1\}$ is a bounded (modulo 2) sequence number.
Initially, $S=\bot$ and each array entry $A[i]$ has the value $(w_i,\bot,0)$ for some fixed $w_i\in D$.

Now suppose process $p$ calls \update{$i,x$}, and this is $p$'s $j$-th update of the $i$-th component of $V$.
To perform the update, $p$ first writes its ID to $S$ and then it writes the triple $(x,p,j \bmod 2)$ to $A[i]$.

To execute a \scan{}, process $p$ first writes its ID to $S$.
Then it performs a collect (i.e., it reads all entries of $A$) to obtain a \emph{view} $a[1\dots M]$, and another collect to obtain a second view $a'[1\dots M]$.
Finally, the process reads $S$.
If $S$ does not contain $p$'s ID or if the views $a$ and $a'$ obtained in the two collects differ, then $p$ starts its \scan{} over; otherwise it returns view $a$.

Obviously \update{} is wait-free and has step complexity $O(1)$.
If process $p$ runs alone for at most $4m+3$ steps of its \scan{} operation, it performs a write to $S$ following by two collects and a read of $S$. Since $p$ runs alone collects are the same and $p$ reads its own ID from $S$ and it must terminate.
Hence solo step complexity of \scan{} is $O(M)$.

To prove linearizability, we use the following linearization points:
Each \update{$i,x$} operation linearizes at the point when the calling process writes to $A[i]$, and each \scan{} operation that terminates linearizes at the point just before the calling process performs its last collect during its \scan{}.
(We don't linearize pending \scan{} operations.)

Consider a \scan{} operation by process $p$ which returns the view $a=a[1\dots M]$.
Let $t$ be the point when that \scan{} linearizes, i.e., just before $p$ starts its last collect.
To prove linearizability it suffices to show that $A=a$ at point $t$.

For the purpose of a contradiction assume that this is not the case, i.e., there is an index $i\in\{1,\dots,M\}$ such that at time $t$ the triple stored in $A[i]$ is not equal to $a[i]$.
Let $t_1$ and $t_2$ be the points in time when $p$ reads the value $(w,q,b)=a[i]$ from $A[i]$ during its penultimate and ultimate collect, respectively.
Then $t_1<t<t_2$.
Since $A[i]\neq (w,q,b)$ at time $t$ but $A[i]=(w,q,b)$ at times $t_1$ and $t_2$, process $q$ writes $(w,q,b)$ to $A[i]$ at some point in the interval $(t,t_2)\subseteq (t_1,t_2)$.
Since $p$ does not write to $A$ during its \scan{}, this implies $q\neq p$.

First suppose $q$ writes to $A[i]$ at least twice during $(t_1,t_2)$.
Each such write must happen during an \update{} operation by $q$.
Since each \update{} operation starts with a write to $S$, $q$ writes its ID to $S$ at least once in $(t_1,t_2)$.
But since the penultimate collect of $p$'s \scan{} starts before $t_1$ and the ultimate collect finishes after $t_2$, $S$ cannot change in the interval $(t_1,t_2)$, which is a contradiction.

Hence, suppose $q$ writes to $A[i]$ exactly once in $(t_1,t_2)$; in particular it writes the triple $(w,q,b)$ to $A[i]$ at some point $t^\ast\in (t_1,t_2)$.
Recall that each time $q$ writes to $A[i]$ it alternates the bit it writes to the third component.
Hence, at no point in $[t_1,t^\ast]$ the second and third component of $A[i]$ can have value $q$ and $b$.
In particular, $A[i]\neq (w,q,b)$ at point $t_1$, which is a contradiction.


\section{Obstruction Freedom vs.\ Randomized Wait-Freedom}
\label{sec:wait-free}
In this section, we present a simple technique that transforms any deterministic obstruction-free algorithm into a randomized one that is equally space efficient and is randomized wait-free against the oblivious adversary.
Moreover, if
the solo step complexity of the deterministic algorithm is $b$,
then the randomized algorithm guarantees that any process finishes after a number of steps that is bounded by a polynomial function of $n$ and $b$.
Precisely, the process finishes in $O\big(b(n + b) \log(n/\delta)\big)$ steps, with probability at least $1-\delta$, as stated in Theorem~\ref{thm:randomized-step-bound}.

A naive approach is the following:
Whenever a process is about to perform a shared memory step in the algorithm, it can flip a coin, and with probability 1/2 it performs the step of the algorithm (called ``actual" step), while with the remaining probability it executes a ``dummy'' step, e.g., reads an arbitrary register.
Suppose the solo step complexity of an obstruction-free algorithm is $b$.
Any execution of length $b n$ (i.e., where exactly $b n$ shared memory steps are performed) must contain a process that executes at least $b$ steps, and with probability at least $1/2^{b n}$ that process executes $b$ actual steps while all other processes execute just dummy steps.
Then during an execution of length $c\cdot b\cdot n\cdot 2^{b n}$ some process runs unobstructed for at least $b$ actual steps with probability $1-1/e^c$.
Hence, the algorithm is randomized wait-free.
This naive transformation yields exponential expected step complexity.

In order to improve the expected step complexity, processes use a biased coin to decide whether to take a larger number of consecutive ``dummy'' or ``actual'' steps.
Precisely, every process $p$ tosses a biased coin before its first step, and also again every $b$ steps.
The outcome of each coin toss is heads with probability $1/n$ and tails with probability $1-1/n$, independently of other coin tosses.
If the outcome of a coin toss by $p$ is heads, then in its next $b$ steps, $p$ executes the next $b$ steps of the given deterministic algorithm; if the outcome is tails then the next $b$ steps of $p$ are \emph{dummy} steps, e.g., $p$ repeatedly reads some shared register.

%
\subsection*{Proof of Theorem~\ref{thm:randomized-step-bound}}
We show that the randomized algorithm described above has the properties specified in Theorem~\ref{thm:randomized-step-bound}.

Let $\sigma = (\pi_1,\pi_2,\ldots)$, where $\pi_i\in\PP$, be an arbitrary schedule determining an order in which processes take steps.
We assume that $\sigma$ is fixed before the execution of the algorithm, and in particular before any process tosses a coin.
For technical reasons we assume that after a process finishes it does not stop, but it takes \emph{no-op} steps whenever it is its turn to take a step according to $\sigma$.
Also the process continues to toss a coin every $b$ (no-op) steps; the outcome of this coin toss has no effect on the execution, but is used in the analysis.

We start with a sketch of the proof.
We sort processes by increasing order in which they are scheduled to take their $(\lambda b)$-th step in $\sigma$, for some $\lambda = \Theta\big((n+b)\log(n/\delta)\big)$.
Let $p_i$ denote the $i$-th process in this order.
We focus on process $p_1$ first.
We consider $\lambda$ disjoint \emph{blocks} of $\sigma$, where the $\ell$-th block, for $1\leq \ell\leq\lambda$, starts with the first step of $p_1$ after its $\ell$-th coin toss, and finishes with the last step of $p_1$ before its next coin toss.
Let $m_\ell$ denote the number of steps contained in block $\ell$;
then $\sum_\ell m_\ell\leq n\lambda b$ by $p_1$'s definition.
Further, the number of coin tosses that occur in block $\ell$ is easily seen to be at most $O(m_\ell/b + n)$.
These coin tosses, plus at most $n$ additional coin tosses preceding the block (one by each process), determine which of the steps 
in the block are actual steps and which are dummy.
If all these coin tosses by processes other than $p_1$ return tails, we say that the block is \emph{unobstructed} (for $p_1$).
Such a block does not contain any actual steps by any processes $p\neq p_1$.
It follows that the probability that block $\ell$ is unobstructed is at least $(1-1/n)^{O(m_\ell/b + n)}$.
The expected number of unobstructed blocks is then $\sum_\ell (1-1/n)^{O(m_\ell/b + n)}$, and we show that this is $\Omega(\lambda)$ using that $\sum_\ell m_\ell\leq n\lambda b$.
Further, we show that this $\Omega(\lambda)$ bound on the number of unobstructed blocks holds also with high probability.
This would follow easily if for different blocks the events that the blocks are unobstructed were independent; but they are not, as they may depend on the outcome of the same coin toss.
Nevertheless the dependence is limited, as each coin toss affects steps in at most $b$ different blocks and each block is affected by at most $O(n)$ coin tosses on average.
To obtain the desired bound we apply a concentration inequality from~\cite{McDiarmid1998}, which is a refinement of the standard method of bounded differences.
Having established that $\Omega(\lambda)$ blocks are unobstructed, it follows that the probability that $p_1$'s coin toss comes up heads at the beginning of at least one unobstructed block is $1-(1-1/n)^{\Omega(\lambda)} = 1-e^{-\Omega(\lambda/n)} \geq 1-\delta/n$ for the right choice of constants.
Hence with at least this probability, $p_1$ finishes after at most $\lambda b$ steps.

Similar bounds are obtained also for the remaining processes:
We use the same approach as above for each $p_i$, except that in place of $\sigma$  we use the schedule $\sigma_i$  obtained from $\sigma$ by removing all instances of $p_j$ except for the first $\lambda b$ ones, for all $1\leq j<i$.
We conclude that with probability $1-\delta/n$, $p_i$ finishes after taking at most $\lambda b$ steps, assuming that each of the processes $p_1,\ldots,p_{i-1}$ also finishes after at most $\lambda b$ steps.
The theorem then follows by applying a union bound. 

Next we give the detailed proof.
Let $\lambda = \beta(n+b)\ln (n/\delta)$, for a constant $\beta>0$ to be fixed later.
Let $p_1,\ldots,p_k$ be all processes that have at least $\lambda b$ steps in schedule $\sigma$, listed in the order in which they execute their $(\lambda b)$-th step.
Let $\sigma_i$, for $1\leq i\leq k$, be the schedule obtained from $\sigma$ after removing all instances of $p_{j}$ except for the first $\lambda b$, for all $1\leq j< i$.
For each $1\leq i\leq k$, we identify $\lambda$ disjoint blocks of $\sigma_i$, where for $1\leq\ell\leq\lambda$, the $\ell$-th block, denoted $\sigma_{i,\ell}$, starts with $p_i$'s step following its $\ell$-th coin toss, 
and finishes after the last step of $p_i$ before its $(\ell+1)$-th coin toss.
By $|\sigma_{i,\ell}|$ we denote the number of steps contained in $\sigma_{i,\ell}$.
We have
$
    \sum_\ell |\sigma_{i,\ell}| \leq n\lambda b,
$
because blocks $\sigma_{i,1},\ldots,\sigma_{i,\lambda}$ contain in total $\lambda b$ steps of each of the processes $p_1,\ldots,p_i$, and fewer than $\lambda b$ steps of each of the remaining processes.

Observe that if $p_i$ has not finished before block $\sigma_{i,\ell}$ begins, and if $p_i$'s coin toss before block $\sigma_{i,\ell}$ returns heads, then $p_i$ is guaranteed to finish during $\sigma_{i,\ell}$ if all other steps by non-finished processes during  $\sigma_{i,\ell}$ are dummy steps.

We say that a coin toss \emph{potentially obstructs} $\sigma_{i,\ell}$ if it is  performed by a process $p\neq p_i$, and at least one of the $b$ steps by $p$ following that coin toss takes place during $\sigma_{i,\ell}$.
This step will be an actual step only if the coin comes up heads
(and $p$ has not finished yet).
We say that block $\sigma_{i,\ell}$ is \emph{unobstructed} if all coin tosses that potentially obstruct this block yield tails.
The number of coin tosses that potentially obstruct $\sigma_{i,\ell}$ is bounded by
$
    |\sigma_{i,\ell}|/b + 2n,
$
because if process $p\neq p_i$ takes $s>0$ steps in $\sigma_{i,\ell}$, then the coin tosses by $p$ that potentially obstruct $\sigma_{i,\ell}$ are the at most $\ceil{s/b}$ ones  that  take place during $\sigma_{i,\ell}$, plus at most one before $\sigma_{i,\ell}$.

It follows that the probability that $\sigma_{i,\ell}$  is unobstructed is at least $(1-1/n)^{|\sigma_{i,\ell}|/b + 2n}$.
Thus the expected number of unobstructed blocks among $\sigma_{i,1},\ldots,\sigma_{i,\lambda}$ is at least
$
    \sum_{\ell}(1-1/n)^{|\sigma_{i,\ell}|/b + 2n}.
$
Using now that $\sum_\ell |\sigma_{i,\ell}| \leq n\lambda b$, and that $(1-1/n)^{x+2n}$ is a convex function of $x$, we obtain that the previous sum is minimized when all $\lambda$ blocks have the same size, equal to $n b$.
Thus, the expected number of unobstructed blocks is at least
\[
    \sum_{1\leq \ell\leq\lambda}(1-1/n)^{|\sigma_{i,\ell}|/b + 2n}
    \geq
    \lambda  (1-1/n)^{(nb)/b + 2n}
    \geq
    \lambda (1-1/n)^{3n}
    >
    \lambda /4^3
    =
    \lambda/64,
\]
where for the last inequality we used that $(1-1/n)^n\geq 1/4$, when $n\geq 2$.

Next we use the following result to establish a lower bound on the number of unobstructed blocks with high probability.
This result is a special case of~\cite[Theorem~3.9]{McDiarmid1998}, which is an extension to the standard method of bounded differences.

\begin{theorem}
    \label{thm:concentration}
Let $X_1,\ldots,X_\kappa$ be independent 0/1 random variables such that $\Pr(X_j = 1) = \rho$, for $1\leq j\leq \kappa$.
Let $f$ be a bounded real-valued function defined on $\{0,1\}^\kappa$, such that
$
    |f(x)-f(x')|\leq c_j,
$
whenever vectors $x,x'\in \{0,1\}^\kappa$ differ only in the $j$-the coordinate.
Then for any $t>0$,
\[
    \Pr\big(|f(X_1,\ldots,X_\kappa) - \Exp[f(X_1,\ldots,X_\kappa)]| \geq t \big)
    \leq
    2e^{-\frac{t^2}{2\rho\sum_{j} c_j^2 + 2t\max_j\{c_j\}/3}}.
\]
\end{theorem}

Let the 0/1 random variables $X_1,X_2,\ldots$ denote the outcome of the coin tosses that potentially obstruct at least one of the blocks $\sigma_{i,1},\ldots,\sigma_{i,\lambda}$: $X_j = 1$ if the $j$-th of those coin tosses is heads, and $X_j = 0$ otherwise.
Then, $\Pr(X_j=1) = 1/n$.
Let $f(X_1,X_2,\ldots)$ be the number of unobstructed blocks.
We showed above that $\Exp[f(X_1,X_2,\ldots)]\geq\lambda/64$.
Further, we observe that flipping the value of $X_j$ can change the value of $f$ by at most the number of blocks that $X_j$ potentially obstructs; let $c_j$ denote that number.
Then, $\max_j{c_j} \leq b$.
Finally, since each block $\sigma_{i,\ell}$ is potentially obstructed by at most $|\sigma_{i,\ell}|/b + 2n$ coin tosses,
\[
    \sum_j c_j
    \leq
    \sum_{1\leq\ell\leq\lambda} (|\sigma_{i,\ell}|/b + 2n)
    =
    \sum_{1\leq\ell\leq\lambda} |\sigma_{i,\ell}|/b + 2n\lambda
    \leq
    3n\lambda,
\]
Thus,
$
    \sum_j c_j^2
    \leq
    \sum_j (c_jb)
    \leq
    3n\lambda b.
$
Applying now Theorem~\ref{thm:concentration} for $t = \lambda/128 \leq \Exp[f(X_1,X_2,\ldots)]/2$ gives
\[
    \Pr\big(f(X_1,\ldots,X_n) \leq t \big)
    \leq
    \Pr\big(\Exp[f(X_1,X_2,\ldots)] - f(X_1,\ldots,X_n) \geq t \big)
    \leq
    2e^{-\frac{t^2}{6b\lambda + 2tb/3}}.
\]
Substituting $t = \lambda/128$ and $\lambda = \beta (n+b)\ln(n/\delta)$, and letting $\beta = 3(6\cdot 128^2 + 2\cdot 128/3) $ yields
$\Pr
\big(f(X_1,\ldots,X_n) \leq \lambda/128 \big)\leq 2e^{-3\ln(n/\delta)} \leq \delta/(2n)$, for $n\geq 2$.
Thus, with probability at least $1-\delta/(2n)$ at least $\lambda/128$ of the blocks $\sigma_{i,1},\ldots,\sigma_{i,\lambda}$ are unobstructed.
The probability that $p_i$ tosses heads before at least one unobstructed block is then at least
\[
    \big(1-\delta/(2n)\big)\cdot\big(1-(1-1/n)^{\lambda/128}\big).
\]
Since $1-(1-1/n)^{\lambda/128}\geq 1-e^{\lambda/(128n)}> 1- \delta/(2n)$, the above probability is at least $\big(1-\delta/(2n)\big)^2\geq 1-\delta/n$.

We have thus far established that for any $1\leq i\leq k$, with probability at least $1-\delta/n$ process $p_i$ finishes after at most $\lambda b$ steps \emph{under schedule $\sigma_i$}.
However, schedules $\sigma$ and $\sigma_i$ yield identical executions if each of the processes $p_1,\ldots,p_{i-1}$ finishes after executing no more than $\lambda b$ steps (the executions are identical assuming the same coin tosses in both executions).
Then, by the union bound, the probability that all processes $p_i$ finish after executing no more than $\lambda b$ steps each is at least $1-n\cdot\delta/n = 1-\delta$.
This concludes the proof of Theorem~\ref{thm:randomized-step-bound}.

\section{Time and Space Efficient Randomized Test-and-Set}\label{sec:randomizedTAS}

In this section, we present a new randomized TAS algorithm that has the properties stated in Theorem~\ref{thm:randomized}.
In particular, it uses a logarithmic number of registers, and has almost constant, $O(\log^\ast n)$, expected step complexity against an oblivious adversary.
The algorithm combines a known randomized TAS construction~\cite{GW2012b}, with the (deterministic) obstruction-free TAS algorithm from Section~\ref{sec:deterministicTAS}, which is turned it into a randomized one by applying the technique of Section~\ref{sec:wait-free}.

We start by observing that since the solo step complexity of the obstruction-free TAS algorithm in Section~\ref{sec:deterministicTAS} is $b = \Theta(\log n)$ (Theorem~\ref{thm:main}), the technique from Section~\ref{sec:wait-free} can be applied.
This yields a randomized TAS algorithm that uses $O(\log n)$ bounded registers, where every process finishes its \Competition{} method after at most $O(n\log^2 n)$ steps, both in expectation and with probability $1-O(1/n^c)$, for any constant $c>0$ (by Theorem~\ref{thm:randomized-step-bound}).
This step complexity is of course much larger than the nearly constant complexity we want to achieve.

Next we give an overview of the randomized TAS algorithm from~\cite{GW2012b} that we will use.
This algorithm has the desired step complexity, but requires (at least) a linear number of registers rather than logarithmic.
To simplify exposition we consider the equivalent weak leader election algorithm rather than the TAS algorithm (see Theorem~\ref{thm:TAS}).
The algorithm uses a chain of $n$ sifter objects $S_1,\ldots,S_n$, alternating with $n$ splitter objects $P_1,\ldots,P_n$, and a chain of $n$ 2-process weak leader election objects $L_n,L_{n-1},\ldots,L_1$.
(A splitter object supports a single operation, \split{}, which returns \win, \lose, or \continue, such that at most one process wins, not all processes lose, and not all continue.)

A process $p$ starts by invoking the \compete{} method of the first sifter object, $S_1$.
If $p$'s invocation of \compete{} in some sifter $S_i$ returns \lose, then $p$ immediately loses in the weak leader election algorithm.
Otherwise, after $p$ wins in $S_i$, it executes the \split{} method of $P_i$:
if this method returns \lose, then $p$ loses immediately, as before;
if it returns \continue, $p$ invokes the \compete{} method of the next sifter, $S_{i+1}$; while if \split{} returns \win, $p$ switches to the chain of 2-process weak leader election objects.
In the last case, $p$ tries to win the 2-process weak leader elections in $L_i,L_{i-1},\ldots,L_1$, in this order.
If $p$ succeeds, it wins the weak leader election algorithm; otherwise it loses, as soon as it loses for the first time in some 2-process weak leader election.

The correctness of the algorithm above follows easily from the next observations.
If exactly one process invokes the \split{} method of a splitter $P_i$, then this invocation returns \win, while if there are $\kappa > 1$ invocations then at least one returns \win or \continue, and no more than $\kappa-1$ return \continue.
This implies that not all processes lose, and that at most $n-i+1$ processes invoke $P_i$'s \split{} method, thus no more than $n$ splitter (or sifter) objects are needed.
The \compete{} method of each 2-process weak leader election object $L_i$ is invoked by no more than two processes: the winner in $L_{i+1}$ (if it exists), and the at most one winner in $P_i$.

In~\cite{GW2012b}, a randomized sifter algorithm is presented that uses $s\geq2$ single-bit registers ($s$ is a parameter), such that the \compete{} method involves just $2$ steps, and if at most $2^s$ processes invoke this method, then at most $O(s)$ of the invocations return \win, in expectation.
Moreover, for $s = 2$, if $k\geq 2$ invocations of the \compete{} method take place, then the expected number of invocations that return \win is at most $k/2 + 1$.
In the following, we will refer to a sifter object implemented by the above algorithm as a \emph{GW-sifter of size $s$}.


The weak leader election algorithm discussed earlier from~\cite{GW2012b}, uses $n$ GW-sifters of size $\log n$ as $S_1,\ldots,S_n$.
This is shown to achieve an expected step complexity of $O(\log^\ast n)$, but requires $\Theta(n\log n)$ registers in total.

Here we propose instead that different types of sifter objects are used, as follows.
The first sifter, $S_1$, is a GW-sifter of size $\log n$, as before.
The next $\ell = \log^2\log n$ objects $S_{2},\ldots,S_{\ell+1}$, are GW-sifters of size $z = 2\log\log n$.
After that, the next $m = \beta \log n$ objects $S_{\ell+2},\ldots,S_{\ell+m+1}$, for $\beta>0$ a sufficiently large constant, are GW-sifters of size~2.
Last, sifter $S_{\ell+m+2}$ is the randomized TAS object obtained by applying Theorem~\ref{thm:randomized-step-bound} to the deterministic TAS algorithm of Theorem~\ref{thm:main}, as discussed at the beginning.
(Recall that any TAS algorithm is also a 1-sifter.)
Objects $S_{i}$, $P_i$, and $L_i$, for $i>\ell+m+2$, are no longer needed.

It is straightforward to verify that this implementation uses $\Theta(\log n)$ registers:
a total of $\log n + z\ell + 2m  = O(\log n)$ registers are used for the sifter objects, and $O(1)$ registers per object suffice for implementing each splitter $P_i$ and randomized 2-process weak leader election object $L_i$~(see \cite{GW2012b}).

We compute now the step complexity of the algorithm.
For the first sifter, the expected number of invocations that return \win is $O(\log n)$. By Markov's inequality the probability that more than $2^{z} = \log^2 n$ invocations return \win is at most $O(\log(n)/\log^2 n) = O(1/\log n)$.

Suppose now that no more than $2^z$ processes invoke the \compete{} method of the $z$-bit sifter $S_{2}$ (which happens with probability $1-O(1/\log n)$ as argued above).
From the analysis in~\cite{GW2012b} it follows that only the first $\mu = O(\log^\ast n)$ of the $\ell$ sifters $S_{2},\ldots,S_{\ell+1}$ are used in expectation.
By dividing this sequence of sifters into $\ell/(2\mu)$ subsequences of $2\mu$ sifters, and applying  Markov's inequality to each, we obtain that the probability all the $\ell$ sifters are used is $1/2^{\ell/(2\mu)} = o(1/\log n)$.
Thus, only with probability $o(1/\log n)$ is sifter $S_{\ell + 2}$ used.

For each sifter $S_{i}$, for $i \in\{\ell + 2,\ldots,\ell + m+1\}$,
we have that if $S_i$'s \compete{} method is invoked $\kappa>1$ times then at most $\kappa/2+1$ invocations return \win in expectation, thus at most $\kappa/2$ processes invoke the \compete{} method of the next sifter, $S_{i+1}$ (because at least one invocation of $P_{i}$'s \split{}  does not return \continue).
Then by Markov's inequality, at most $2\kappa/3$ processes invoke $S_{i+1}$'s \compete{} method, with probability at least 1/4. 
Therefore, no more than $4\log_{3/2}n < 8\log n$ of the $m$ sifters are used in expectation.
By a standard Chernoff bound argument, the probability that the last sifter, $S_{\ell+m+2}$, is used can be made smaller than $n^{-\beta'}$, for any constant $\beta'>0$, by choosing a sufficiently large constant $\beta$.

Combining the above we obtain that the expected number of sifters, other than the last one, that are used is at most
\[
    1 + O(\log^\ast n) + \left(O\left(1/\log n\right) + o(1/\log n)\right)\cdot m = O(\log^\ast n),
\]
where the first term on the left accounts for $S_1$, the second for the expected number of sifters used among the next $\ell$ sifters, and the third term for the expected number of sifters used among the subsequent $m$ sifters, where  factor $O(1/\log n) + o(1/\log n)$ is the probability that either more than $2^z$ processes use $S_2$ or some process uses $S_{\ell + 2}$.
Using also that the last sifter is used with probability at most $n^{-\beta'}$ and has step complexity $O(n\log^2n)$, we obtain that the expectation for the maximum number of steps by any process is at most
\[
    O(\log^\ast n) +  n^{-\beta'} \cdot O(n\log^2 n) = O(\log^\ast n),
\]
for $\beta'>1$.
Further the probability that the maximum number of steps is $O(\log n)$ is $1-O(n^{-\beta'})$, which follows from the probability that the last sifter, $S_{\ell+m+2}$, is not used.

\clearpage
\bibliographystyle{plain}
\bibliography{tasbib-lowercase}

\end{document}